\documentclass[a4paper,UKenglish]{lipics_v2}

\usepackage{microtype}

\usepackage{amsmath,amsfonts,amssymb,amsthm}
\usepackage{xspace}
\usepackage{paralist}
\usepackage{graphicx}
\usepackage{array}
\usepackage{multirow}

\usepackage[ruled,vlined,linesnumbered]{algorithm2e}
\SetFuncSty{textsc}
\SetDataSty{textit}
\SetCommentSty{textrm}
\SetKwFunction{Extract}{ExtractFirstDigit}
\SetKwData{Forward}{forward}
\SetKw{True}{true}

\newcommand{\Frechet}{Fr\'echet\xspace}
\newcommand{\dist}[2]{\ensuremath{d\pth{#1,#2}}}
\newcommand{\distFr}[2]{\ensuremath{d_F\pth{#1,#2}}}
\newcommand{\distDTW}[2]{\ensuremath{d_{\text{DTW}}\pth{#1,#2}}}

\newcommand{\distAFrx}[3]{\ensuremath{d_{#3, \text{aF}}\pth{#1,#2}}}
\newcommand{\distAFr}[2]{\ensuremath{d_{w, \text{aF}}\pth{#1,#2}}}
\newcommand{\distADTWx}[3]{\ensuremath{d_{#3, \text{aDTW}}\pth{#1,#2}}}
\newcommand{\distADTW}[2]{\ensuremath{d_{w, \text{aDTW}}\pth{#1,#2}}}

\newcommand{\distSFrx}[3]{\ensuremath{d_{#3, \text{sF}}\pth{#1,#2}}}
\newcommand{\distSFr}[2]{\ensuremath{d_{w, \text{sF}}\pth{#1,#2}}}

\newcommand{\distSDTW}[2]{\ensuremath{d_{w, \text{sDTW}}\pth{#1,#2}}}

\newcommand{\Hash}{\ensuremath{\mathcal {H}}\xspace}
\newcommand{\Prob}[2]{Pr_{#2}\pth{#1}}
\newcommand{\CurveSet}{\ensuremath{\mathcal S}\xspace}
\newcommand{\Curves}[1]{\ensuremath{\Delta^{#1}}\xspace}
\newcommand{\TraversalSet}{\ensuremath{\mathcal T}\xspace}
\newcommand{\subseq}[2]{\ensuremath{\widehat{#1}_{#2}}}
\newcommand{\Partition}[2]{\ensuremath{\PartitionFunc{#1}\pth{#2}}}
\newcommand{\PartitionFunc}[1]{\ensuremath{\Phi^{#1}}}
\newcommand{\Partitions}{\ensuremath{\mathcal P}\xspace}

\providecommand{\eps}{{\varepsilon}}%

\providecommand{\ceil}[1]{\left\lceil {#1} \right\rceil}

\providecommand{\pth}[2][\!]{#1\left({#2}\right)}
\providecommand{\brc}[1]{\left\{ {#1} \right\}}

\newcommand{\pbrc}[1]{\left[ {#1} \right]}
\renewcommand{\Re}{{\rm I\!\hspace{-0.025em} R}}
\newcommand{\Na}{{\rm I\!\hspace{-0.025em} N}}
\DeclareMathOperator*{\argmin}{arg\,min}
\DeclareMathOperator*{\poly}{poly}

\newcommand{\etal}{\textit{e{}t~a{}l.}\xspace}

\newcommand{\seclab}[1]{\label{sec:#1}}
\newcommand{\secref}[1]{Section~\ref{sec:#1}}

\newcommand{\lemlab}[1]{{\label{lem:#1}}}
\newcommand{\lemref}[1]{Lemma~\ref{lem:#1}}

\newcommand{\thmlab}[1]{{\label{theo:#1}}}
\newcommand{\thmref}[1]{Theorem~\ref{theo:#1}}

\newcommand{\corlab}[1]{\label{cor:#1}}
\newcommand{\corref}[1]{Corollary~\ref{cor:#1}}

\newcommand{\applab}[1]{\label{app:#1}}
\newcommand{\appref}[1]{Appendix~\ref{app:#1}}

\newcommand{\BO}[1]{O\pth{#1}}
\newcommand{\BOM}[1]{\Omega\pth{#1}}

\title{Locality-sensitive hashing of curves\footnote{Driemel  has been supported by NWO Veni project ``Clustering time series and trajectories (10019853)''. Silvestri has been supported by the European Research Council project ``Scalable Similarity Search'' (no. 614331) and by MIUR of Italy under project AMANDA.}}

\author[1]{Anne Driemel}
\author[2]{Francesco Silvestri}
\affil[1]{Department of Mathematics and Computer Science, Eindhoven University of Technology, The Netherlands\\
  \texttt{adriemel@tue.nl}}
\affil[2]{
Department of Information Engineering, University of Padova, Italy\\
Theoretical Computer Science, IT University of Copenhagen, Denmark\\
\texttt{silvestri@dei.unipd.it}}
\authorrunning{A. Driemel and F. Silvestri}

\Copyright{}


\begin{document}
 
\maketitle

\begin{abstract} 
We study data structures for storing a set of polygonal curves in $\Re^d$ such that, given a query curve, we can efficiently retrieve similar curves from the set, where similarity is measured using the discrete \Frechet distance or the dynamic time warping distance.  To this end we devise the first locality-sensitive hashing schemes for these distance measures.  A major challenge is posed by the fact that these distance measures internally optimize the alignment between the curves. We give solutions for different types of alignments including constrained and unconstrained versions.  For unconstrained alignments, we improve over a result by Indyk from 2002~\cite{i-approxnn-02} for short curves. Let $n$ be the number of input curves and let $m$ be the maximum complexity of a curve in the input.  In the particular case where $m \leq  \frac{\alpha}{4d} \log n$, for some fixed $\alpha>0$, our solutions imply an approximate near-neighbor data structure for the discrete \Frechet distance that uses space in $O(n^{1+\alpha}\log n)$ and achieves query time in $O(n^{\alpha}\log^2 n)$ and constant approximation factor.  Furthermore, our solutions provide a trade-off between approximation quality and computational performance: for any parameter $k \in [m]$, we can give a data structure that uses space in $O(2^{2k}m^{k-1} n \log n + nm)$, answers queries in $O( 2^{2k} m^{k}\log n)$ time and achieves approximation factor in $O(m/k)$.

\end{abstract}

\section{Introduction}
\seclab{intro}
We study nearest-neighbor searching for polygonal curves under the discrete
\Frechet distance or the dynamic time warping distance. This problem has
various applications in machine learning, information retrieval and
classification where the recorded instances are curves.
Dynamic time warping has shown to be useful for classification of various types
of data: surgical processes \cite{Forestier2012255}, whale
singing \cite{bmp-ackw-07}, chromosomes \cite{Legrand2008215}, fingerprints
\cite{888711}, electrocardiogram (ECG) frames \cite{1013101}, and vessel
trajectories~\cite{vries2012kernel}.  Originally conceived for
speech recognition, it is now being deployed as universal similarity measure
for time series in the field of data mining. The \Frechet distance is
considered a useful similarity measure for trajectories of moving objects
~\cite{playerpaths2015,gpu-clustering-frechet2015, mdl-frechet-2014, trajectory-corridors2010}.

Indyk and Motwani~\cite{indyk1998approximate, him-12} introduced the idea that
hashing could enable faster nearest-neighbor searching in high-dimensional
Euclidean spaces using a hashing scheme where near points are more likely to
collide than far ones. They showed that such an approach can be used for the
\emph{$(c,r)$-near neighbor problem} which is defined as follows.  Preprocess a
set $S$ of $n$ points into a data structure that answers queries in the
following way: if there exists a point $p \in S$ that lies within distance $r$
from the query point  $q$, then the data structure reports a point $p'\in S$ that
lies within distance $cr$ from $q$. In this paper, we study such locality-sensitive 
hashing schemes for the space of curves.

\subsection{State of the art}

In 2002, Indyk gave a deterministic and approximate near-neighbor data
structure for the discrete \Frechet distance~\cite{i-approxnn-02}. 
This data structure is to date the only result  known for this task and
represents the state of the art. The data structure achieves 
approximation factor $O(\log m + \log\log n)$, where $m$ is
the maximum length of a curve and $n$ is the maximum number of
elements in the data structure. Further, the data structures uses space in 
$O\pth{|X|^{\sqrt{m}} (m^{\sqrt{m}} n)^2}$, where $|X|$ is the size of the
domain on which the curves are defined.  The query time is $O\pth{m^{O(1)}\log n}$.
The data structure precomputes all answers to queries with curves of length
$\sqrt{m}$, leading to a very high space consumption.\footnote{Indyk also
claims (without proof) a slightly different bound using a trade-off parameter
$t \geq 2$: approximation factor $O((\log m + \log\log n)^{(t-1)})$,  space
$O\pth{\pth{m^2|X|}^{tm^{1/t}} n^{2t} }$ and query time $(m+\log n)^{O(t)}$. 
The space bound decreases at the cost of approximation and query time as soon as $t<\log m$; however, the trade-off disappear for larger values of $t$  since all bounds increase in $t$ as soon as $t\geq \log m$.
} 

In the group of $\ell_p$ distances, the \Frechet distance most resembles the 
$\ell_{\infty}$-distance, which is notoriously hard to embed into a
low-dimensional  $\ell_p$-space, see also the discussion by Indyk in~\cite{i-nes-98}.  
Indyk's data structure for the discrete \Frechet distance
is in fact an extension of his data structure for the $\ell_{\infty}$-distance~\cite{i-nes-98}. 
Any subset of $\ell^d_{\infty}$ can be embedded into the \Frechet 
metric\footnote{In particular, one can use $3d$ vertices to express each
$d$-dimensional vector as a curve on a real line.}~\cite{indmat04}.  
This embedding implies that, unless the strong exponential-time hypothesis
fails, there exists no data structure for near-neighbor searching
under the discrete \Frechet distance that achieves preprocessing time in
$O\pth{n^{2-\eps}\poly{m}}$, query time in $O\pth{n^{1-\eps} \poly{m}}$ for
any $\eps>0$, and approximation factor $c<3$.
(see \appref{clb} for details). 
This suggests that the problem becomes hard for long curves, i.e., $m \in \omega\pth{\log n}$.
Recently, Backurs and Sidiropoulos showed how to embed finite subsets of the
Hausdorff distance into $\ell_{\infty}$ using constant distortion and constant
dimension of the host space~\cite{backurs2016constant}. However, for the
\Frechet distance, no non-trivial embeddings are known, see also the discussion
in~\cite{indmat04}.  It is possible to embed any finite metric space into
$\ell_{p}$, for example, using the embedding due to
Bourgain~\cite{mat-embeddings}. However, the high cost of computing the
embedding makes it unfit for use in a nearest-neighbor data structure.  Another
known approach to proximity searching in metric spaces is to exploit a low
doubling-dimension~\cite{Arya2008,gupta2003bounded}.  However, the doubling
dimension of the \Frechet distance is infinite, even if the metric space is
restricted to curves of constant length~\cite{dks-clust-16}.  Recently
Bartal~\etal~\cite{bartal2014impossible} gave lower bounds for embedding
doubling spaces. Their result implies that a metric embedding of the \Frechet
distance into an $\ell_p$-space would have at least super-constant distortion.
However, as noted earlier, it is not even known how to obtain such an embedding.

In general, there is little known in terms of  data structures for the
\Frechet distance.  The authors are aware of the following few results which were
developed for the classic (continuous) \Frechet distance.  De Berg, Cook and
Gudmundsson~\cite{bcg-ffq-13} study range counting queries for the set
of subcurves that lie within distance $r$ to a query line segment. 
Their data structure uses a partition tree to store compressed subcurves. For
any parameter $n\leq s\leq n^2$, the space used by the data structure is in
$O\pth{s \poly\log n}$. The queries are computed in time in
$O\pth{\frac{n}{\sqrt{s}}\poly\log n }$ and uses a constant approximation
factor. However, the data structure does not support more complex query curves
than line segments.  A second data structure is due to Driemel and
Har-Peled~\cite{dh-jay-12}. This data structure answers queries for the
\Frechet distance of a subcurve to a query curve (the subcurve is specified in
the query). If the queries are line segments, an approximation factor of
$(1+\eps)$ can be achieved with logarithmic query time and linear space.
Unlike the $\ell_{\infty}$-metric, which can be evaluated in time that
is linear in the dimension, evaluating a single \Frechet distance is believed
to take time that is at least roughly quadratic in the complexity of the curves
(the number of vertices) in the worst case~\cite{b-seth-14}.  The high time
complexity can be credited to the fact that the distance measure optimizes over
all possible monotone alignments of the two input sequences. Computing the
discrete \Frechet distance, as well as dynamic time warping, can be solved via
dynamic programming.  In both cases, the naive linear scan leads to $O(nm^2)$
query time for finding the nearest neighbor.
For dynamic time warping (DTW) no data structures exist that give provable
guarantees, however there exist many heuristics, see the work of Rakthanmanon
\etal~\cite{rakthanmanon2012searching}~(and references therein).  Since DTW
does not satisfy the triangle inequality, it cannot be embedded into an
$\ell_p$-space.

\subsection{Our results}

Our first result is a basic LSH scheme for the discrete \Frechet distance,
which  leads to a very efficient LSH with approximation factor that is linear
in the number of curve vertices.  The scheme is described in
\secref{basic:all} and it is surprisingly simple: We snap the curves to a
randomly shifted grid and remove consecutive duplicate vertices.
It turns out that this simple scheme alleviates the alignment problem which
sets the \Frechet distance computation apart from the $\ell_{\infty}$-distance.
Next, we show in \secref{const:all} that it is even possible to get constant
approximation, at the cost of a lower collision probability for near curves.
The second scheme randomly perturbs the vertices of the input curves
independently and snaps the vertices to a fixed grid instead of a randomly
shifted grid.  It is natural to ask if there exists an LSH scheme exhibiting
a full-spectrum trade-off between collision probability and approximation.  We
positively answer to this question  in \secref{improved:all}, with a scheme
based on a random partition  of the input curves,  inspired by Indyk's data
structure~\cite{i-approxnn-02}, followed by the application of the basic scheme
to each subsequence independently.\footnote{Indeed, in \thmref{improved:dfd}, the
collision probability for near curves is bounded by $2^{-3M}$ for $K=M$ (using
Stirlings approximation for the binomial coefficient), however when summarizing
our bounds we use the simplified bound from \corref{improved:dfd:bound}.}

All the LSH schemes achieve zero false-positives, meaning that no collisions happen
between far curves. When applied to solve the $(c,r)$-near neighbor
problem, we obtain the results summarized in Table~\ref{results} (see also \secref{lsh}).  It is
interesting to compare our bounds with the state of the art.
The basic scheme of \thmref{basic:dfd} provides a data structure using almost
linear space and $\BO{m \log n}$ query time by allowing a linear approximation
$c=O(m)$. This query time always beats the trivial exact
solution of scanning all input curves for each query, which needs $\BO{nm^2}$
time. In comparison, Indyk's result~\cite{i-approxnn-02} provides a better
approximation when $m=\BOM{\log \log n}$ but it uses exponential space and
slightly higher query time. More generally, when curves are short  $m=o(\log n)$, 
our basic result provides a good alternative to Indyk's result due to the improved space.
In the particular case where $m \leq  \frac{\alpha}{4d} \log n$ for some fixed
$\alpha>0$, we can  answer queries using a constant approximation factor in
$O(n^{\alpha}\log^2 n)$ time and using $O(n^{1+\alpha}\log n)$ space, using
\thmref{const:dfd}.  When curves have constant complexity, the basic
LSH gives the first efficient data structure with constant approximation. 

We then address LSH for the discrete \Frechet distance under alignment constraints in
\secref{constraints}. It is natural to constrain the alignments between 
the vertices of a curve: this preserves important characteristics of the input
curves and it also reduces the actual time to compute the distance between
curves (see e.g.,~\cite{RatanamahatanaK05,Keogh05}).  We target the anchored
and bounded speed constraints that require, respectively, a vertex to be
aligned with at most $w$ vertices or to be aligned with vertices whose indices 
differ by at most $w/2$, for a suitable parameter $w\geq 1$ (for formal
definitions see \secref{prelim:constraints}).  Our scheme provides the first
data structures for the $(c,r)$-near neighbor problem 
with alignment constraints.  
Further, they exhibit a bi-criteria approximation: it is possible to reduce space and query
time with a weaker approximation on the distance but also  on the alignment
parameter $w$.  Bounds are summarized in Table~\ref{results}.

In \secref{ex:dtw}, we study which one of our schemes work for DTW.  We
show that the basic LSH applies to DTW with the same linear approximation,
space and query bounds of the discrete \Frechet distance.  In contrast, the
techniques to improve the approximation factor under the \Frechet distance do
not provide improvements for DTW.
The LSH schemes for constrained distances also yields  linear approximation for
DTW distance, but maintains the trade-off between space/query time and the
approximation on the alignment parameter $w$.

We remark that a data structure for the $(c,r)$-approximate near neighbor problem can be used
as a building block for solving the $c$-approximate nearest neighbor problem. We refer
to~\cite{nn-book} for more details.

\begin{table}[t]
\centering
\caption{Our approximate near-neighbor data structure results for the discrete
\Frechet distance in comparison with the result by Indyk, assuming $d=O(1)$ for simplicity.
The first four rows refer to the standard discrete \Frechet distance $d_F$, while the last two rows $d_{w, \text{aF}}$ and $d_{w, \text{sF}}$ refer to the anchored and speed constraints  respectively.
The input consists of
$n$ polygonal curves in $\Re^d$, each of complexity at most $m$. The
corresponding query results are achieved with high probability. 
The parameters $k\geq 1$ and $\ell\geq 1$ trade-off space/query time and approximation,   and parameter $w$  constrains the possible alignments.
The first entry in 
bi-criteria $ \pth{ \cdot, \cdot}$ denotes the distance approximation, while the second is the alignment approximation.
}
\label{results}
\setlength{\tabcolsep}{1pt}
\begin{tabular}{|l|l|l|l|l|}
\hline
&       Space  & Query time & Approximation & Reference \\
\hline
\multirow{4}{*}{$d_F$} &
$O\pth{|X|^{\sqrt{m}} (m^{\sqrt{m}} n)^2}$ & $O\pth{m^{O(1)}\log n}$ & $O(\log m + \log\log n)$ & \cite{i-approxnn-02} \\
& $O(n \log n + nm)$      &  $O(m \log n)$          &  $O(m)$             & Thm. \ref{theo:basic:dfd} \\
& $O(2^{4md}  n\log n + nm)$      &  $O( 2^{4md} m \log n)$     &  $O(1)$             & Thm. \ref{theo:const:dfd} \\
&  $O(2^{2k} m^{k-1} n \log n +nm) $  &  $O(2^{2k} m^{k} \log n)$   &  $O\pth{{m}/{k}}$   &  Cor. \ref{cor:improved:dfd:bound}         \\
\hline
$d_{w, \text{aF}}$ &  $O\pth{ \pth{\sqrt{2}w}^{2m/\ell} n\log n+nm} $     &    $O\pth{\pth{\sqrt{2}w}^{2m/\ell} m \log n}$      & bi-criteria $ \pth{4d^{\frac{3}{2}} \ell, 2\ell-2}$   & Thm. \ref{theo:anchored}          \\
\hline
$d_{w, \text{sF}}$ &  $O\pth{ \pth{\sqrt{2}w\ell}^{2m/\ell} n\log n + nm} $     &    $O\pth{\pth{\sqrt{2}w\ell}^{2m/\ell}m  \log n}$      & bi-criteria $ \pth{4d^{\frac{3}{2}}\ell, \ell}$   & Thm. \ref{theo:speed}          \\
\hline
\end{tabular}
\end{table}

\section{Preliminaries}
\seclab{prelim}
\subsection{Distance measures for curves}
A \emph{time series} (or \emph{trajectory})\footnote{Usually, these are referred
to as time series when $d=1$ and trajectories when $d>1$.} is a series
$(p_1,t_1),\ldots ,(p_m,t_m)$ of measurements $p_i$ of a signal taken at times
$t_i$. We assume $0=t_1<t_2<\ldots <t_m=1$ and $m$ is finite. A time series may
be viewed as a continuous function $P: [0,1] \rightarrow \Re^d$ by linearly
interpolating $p_1,\dots,p_m$ in order of $t_i$, $i=1,\ldots m$. We obtain a
polygonal curve with \emph{vertices} $p_1=P(t_1),\dots, p_m=P(t_m)$ and
segments between $p_i$ and $p_{i+1}$ called \emph{edges} $\overline{p_i
p_{i+1}}=\{xp_i+(1-x)p_{i+1}|x\in [0,1]\}$. We will simply refer to $P$ as a
\emph{curve}. We denote the space of all curves in $\Re^d$ with $\Curves{d}$. 

We now recall the definitions of discrete \Frechet distance and of the dynamic
time warping distance between two curves.  To this end we define the concept of
traversal.  Given two polygonal curves $P=p_1,\dots,p_{m_1}$ and
$Q=q_1,\dots,q_{m_2}$, a \emph{traversal} 
\[T=\brc{(i_1,j_1), (i_2,j_2),\dots,(i_{\ell},j_{\ell})}\] 
is a sequence of pairs of indices referring to a \emph{pairing} of
vertices from the two curves with the following properties:
\begin{compactenum}[(i)]
\item $i_1=1$, $j_1=1$, $i_{\ell}=m_1$, and $j_{\ell}=m_2$
\item $\forall (i_k,j_k) \in T: (i_{k+1}-i_k) \in \brc{0,1} \wedge (j_{k+1}-j_{k}) \in \brc{0,1} $.
\item $\forall (i_k,j_k) \in T: (i_{k+1}-i_k) + (j_{k+1}-j_{k}) \geq 1 $.
\end{compactenum} 
Intuitively, one can think of the traversal as a prescribed schedule for
simultaneously traversing the two curves, starting at the first vertex of each
curve, in every step the traversal advances by one vertex, either on one of the curves, 
or on both curves simultaneously, finally the traversal has to end at the last
vertices of the two curves. 

We consider the maximum distance of two vertices paired by a traversal as the cost incurred 
by this traversal. Let $\TraversalSet$ be the set of possible traversals for two
curves $P$ and $Q$, then the \Frechet distance corresponds to the minimal cost
of a traversal of the two curves. Likewise, if we define the cost of a traversal as the
sum of distances between paired vertices, then the traversal with minimum cost
corresponds to the dynamic time warping distance. 

\begin{definition}\label{def:fr}
Let $\TraversalSet$  be the set of possible traversals for two
curves $P$ and $Q$. 
The \emph{discrete \Frechet distance $\distFr{P}{Q}$} between curves $P$ and $Q$ is  defined as 
\begin{align*}
\distFr{P}{Q} &= \min_{T \in \TraversalSet} \max_{(i_k,j_k) \in T} \| p_{i_k} - q_{j_k}\|.
\end{align*}
\end{definition}

\begin{definition}\label{def:dtw}
Let $\TraversalSet$ be the set of possible traversals for two
curves $P$ and $Q$.
The \emph{dynamic time warping (DTW) distance $\distDTW{P}{Q}$}  between curves $P$ and $Q$ is defined as 
\begin{align*}
\distDTW{P}{Q} = \min_{T \in \TraversalSet} \sum_{(i_k,j_k) \in T} \| p_{i_k} - q_{j_k}\|.
\end{align*}
\end{definition}
 
The discrete \Frechet distance satisfies the triangle inequality and is a pseudo-metric. 
This is not true for the DTW distance, since it does not satisfy the triangle inequality.

We refer to a traversal realizing the distance of two curves as an
\emph{optimal traversal}.  We can interpret a traversal as the edges of a
bipartite graph where the nodes are the vertices of the two curves and the
edges connect the pairs. 
The following simple lemma holds for all distance measures.  As a consequence,
we  assume in the paper that an optimal traversal consists of disconnected
stars, that we call \emph{components}.  

\begin{lemma}\lemlab{components}
For any two curves $P=p_1,\dots,p_{m_1}$ and  $Q=q_1,\dots,q_{m_2}$, there always
exists an optimal traversal $T$ with the following two properties:
\begin{compactenum}[(i)]
\item $T$ consists of at most $m=\min\{m_1,m_2\}$ disconnected components.
\item Each component is a star, i.e., all edges of this component share a common vertex.
\end{compactenum}
\end{lemma}

\begin{proof}
The first part is immediate, since we can charge each component to a vertex of
the shorter curve that is contained in it.  To see the second part of the claim, assume
for the sake of contradiction that an optimal traversal has the pairs
$(i,j),(i,j+1)(i+1,j+1)$ for some $i,j$ (or the symmetric configuration
$(i,j),(i+1,j)(i+1,j+1)$). In this case, the middle pair $(i,j+1)$ can be removed without
increasing the cost and without invalidating the traversal properties. We can
apply this reasoning repeatedly until each component is a star.
\end{proof}

\subsection{Distances measures with constraints}
\seclab{prelim:constraints}

\emph{Anchored distances.}
A traversal $T$ is said \emph{$w$-anchored traversal} if each vertex  is paired with a vertex at distance at most $w/2$ (for simplicity we assume $w$ to be even): namely,  $|i-j|\leq w/2$ for each  $(i,j)\in T$. Parameter $w$ is called the \emph{width} of the traversal. 
Such a traversal exists only if $|m_1-m_2|\leq w/2$, otherwise there would be unpaired vertices (e.g., the last vertex of the longest curve).
For two curves $P$ and $Q$ with lengths satisfying $|m_1-m_2|\leq w/2$, we define the \emph{$w$-anchored discrete \Frechet distance $\distAFr{P}{Q}$} and \emph{$w$-anchored DTW distance $\distADTW{P}{Q}$} as in Definitions~\ref{def:fr} and~\ref{def:dtw} where 
$\TraversalSet$ is defined as the set of all possible  $w$-anchored traversals. 

\emph{Speed-constrained distances.}
A traversal $T$ is a \emph{$w$-speed traversal} if each vertex is aligned with at most $w$ vertices of the other curve: in other terms,  the bipartite graph representing the traversal has degree at most $w$.
Parameter $w$ is called the  \emph{speed} of the traversal. 
(We  overload the meaning of $w$ since the width and speed parameters play a similar role in our algorithms.)
Note that a $w$-anchored traversal is a $(w+1)$-speed traversal, but the opposite is not necessary true.
A $w$-speed traversal exists only if $1/w \leq m_1/m_2 \leq w$.
For two curves $P$ and $Q$ with lengths satisfying $1/w \leq m_1/m_2 \leq w$, we defined the \emph{$w$-speed discrete \Frechet distance $\distSFr{P}{Q}$} and \emph{$w$-speed DTW distance $\distSDTW{P}{Q}$} as in Definitions~\ref{def:fr} and~\ref{def:dtw} where  $\TraversalSet$ is defined as the set of all possible  $w$-speed traversals.

\subsection{Locality-sensitive hashing}
\seclab{lsh}

We use the notion of asymmetric locality-sensitive hashing (see, e.g.~\cite{Shrivastava14}), defined as follows:

\begin{definition}  Let $\CurveSet$ be the set of curves in $\Re^d$ and let
$d: {\CurveSet} \times {\CurveSet} \rightarrow \Re^{+}$ be a distance measure 
defined on them.  Given real values $r> 0$, $c>1$, $0\leq \alpha_1\leq 1$ and $0\leq \alpha_2 \leq 1$ with $\alpha_1>\alpha_2$, a family \Hash of pairs of hash functions  $(h_1, h_2)$ is called \emph{$(r,c,\alpha_1,\alpha_2)$-sensitive} 
if for any two curves $P,Q \in \CurveSet$
\begin{compactenum}[(i)]
\item if $d(P,Q) \leq r$, then $\Prob{h_1(P)=h_2(Q)}{(h_1,h_2)\in\Hash} \geq \alpha_1$;
\item if $d(P,Q) > cr$, then $\Prob{h_1(P)=h_2(Q)}{(h_1,h_2)\in\Hash} \leq \alpha_2$.
\end{compactenum}
\end{definition}
When $h_1=h_2$, we have the traditional definition of (symmetric) locality-sensitive hashing.
The above scheme is \emph{asymmetric} in the sense that there are two different
schemes and the guarantees only hold for curves $P$ and $Q$ where $P$ was
hashed using the first scheme and $Q$ was hashed using the second scheme.  This
is useful, e.g., if the application of the LSH is a nearest neighbor data
structure, where comparisons only need to be done between input objects and
query objects.

The results  reported in Table~\ref{results} follow by applying  the standard framework for solving the $(c,r)$-near neighbor problem with an
$(r,c,\alpha_1,\alpha_2)$-sensitive hashing scheme \Hash.
For the
sake of completeness, we sketch this process here.\footnote{
We observe that the
LSH schemes presented in this paper have long hash values (curves or array of
curves). However, they can be shortened with traditional hashing (i.e., by
mapping each value in $[0,O(n)]$), that allows for a more efficient search in
the hash tables generated by the LSH.  This technique increases $\alpha_2$ by
an additive $O(1/n)$ term. }
A new family $\Hash'$ of hashing is constructed by concatenating
$k=\max\{1,\log_{\alpha_2}(1/n)\}$ hash functions from $\Hash$, so that the
collision probability of far points is at most $1/n$.  Then,  each point in $S$
is inserted into $L=(1/\alpha_1)^k$ hash tables, each corresponding to a different
randomly chosen hash function from $\Hash'$.
For a query point $q$, the algorithm searches among all points that collide
with $q$ in the $L$ hash tables and stops as soon as a $cr$-near neighbor is
found.  When $\alpha_2>0$, the data structure requires $\BO{n^{1+\rho}+nm}$
memory words and query time $\BO{\Gamma n^{\rho}}$, where $\Gamma=\BOM{m}$ is
the time required for computing the distance between two curves and $\rho=\log
\alpha_1 / \log \alpha_2$.  When $\alpha_2=0$,  the data structure requires
$\BO{n /\alpha_1}$ memory words and query time $\BO{m/\alpha_1}$. Note that in
this case the query time does not include $\Gamma$: the algorithm does
not need to compute  distances between $q$ and points in the buckets since
there are no false positives.  For a given query, the data structures returns
an approximate $cr$-near neighbor with constant probability.  In order to
obtain high probability (i.e., at least $1-1/n$) we repeat the above process
$\log n$ times, leading to $\log n$ different data structures. This increases
space and query time by a $\BO{\log n}$ term.

\section{Linear approximation factor}
\seclab{basic:all}
We first present the basic LSH scheme in \secref{basic}, and then in
\secref{basic:frechet} we analyze its correctness and performance for the
discrete \Frechet distance. The basic LSH has an approximation factor that is linear
in the number of vertices that a curve can have. 

\subsection{Algorithm}
\seclab{basic}

We use a randomly shifted grid in our hashing scheme.
Let the canonical $d$-dimensional grid of resolution $\delta$ be defined as an 
evenly spaced point set in $\Re^d$, as follows: 
\[ G_{\delta} = \brc{ (x_1,\dots,x_d) \in \Re^d ~|~ \forall~ 1 \leq i \leq d ~\exists~ j \in \Na ~:~  x_i = j \cdot \delta }.\]
Consider a family of such grids parametrized by a shift $t$: 
\[ \widehat{G}^t_{\delta} = \brc{ p +  t ~|~ p \in G_{\delta} }. \]
Choosing $t$ uniformly at random from the half-open  hypercube $[0,\delta)^d$ we
obtain a family of randomly shifted grids. 
Let $P \in \Curves{d}$ be a polygonal curve with vertices $p_1,\dots,p_{m}$ and 
let $h^{t}_{\delta}: \Curves{d} \rightarrow \Curves{d}$ be a hash function. 
The curve $h^t_{\delta}(P)$ is defined as the result of the following two-stage construction. 
\begin{compactenum}[(i)]
\item We snap the curve to the grid $\widehat{G}^{t}_{\delta}$. More precisely, 
we replace each vertex $p_i$ with its closest grid point 
$p_i' = \argmin_{q \in \widehat{G}^t_{\delta}} \|p_i - q\|$ to obtain the curve $P'$. 
\item We remove consecutive duplicates in $P'$. 
That is, we remove the vertex $p_i'$ if it is identical to $p_{i-1}'$. 
\end{compactenum}
Let $\Hash^{\texttt{L}}_{\delta}$ be the family of hash functions $h^t_{\delta}$ constructed
this way.

\subsection{Analysis}
\seclab{basic:frechet}

\begin{lemma}\lemlab{dfd:a1}
Let $P, Q \in \Curves{d}$ be two curves with $m_1$ and $m_2$ points, respectively,
and let $m=\min\{m_1,m_2\}$. 
For any $\delta>0$, it holds that  
\begin{align*}
\Prob{h^t_{\delta}(P) = h^t_{\delta}(Q)}{\Hash^{\texttt{L}}_{\delta}} \geq 1- \pth{2 d m \cdot \frac{\distFr{P}{Q} }{\delta}   }.
\end{align*}
\end{lemma}

\begin{proof} 
We bound the probability that $P$ and $Q$ do not hash to the same sequence.
To this end, consider an optimal traversal $T$ of $P$ and $Q$ with respect to
the discrete \Frechet distance. By \lemref{components}, we can assume that 
$|T|\leq m$ and each component is a star. Let $\ell$ denote the number of components of $T$.
For $1 \leq k \leq \ell$ denote with $E_k$ the event that not all vertices of the
$k$-th component are snapped to the same grid point.  This happens only if at
least one pair of vertices is separated in at least one dimension by the
random shift $t$. 

Since the component is a star, there exists a vertex $v$ of either $P$ or $Q$, such that
$v$ is involved in all pairs of $T$ in the $k$-th component. Therefore, all vertices
in this component have distance at most $\distFr{P}{Q}$ to $v$.
Since  $t$ is uniformly distributed in $[0,\delta)^d$, the probability that any pair is
separated along any fixed dimension is $2 \distFr{P}{Q}/\delta$.  
As a consequence, event $E_k$ happens with probability at most $2d\cdot \distFr{P}{Q}/\delta$.

By a union bound over the $\ell$ components in $T$, we have
that the probability of $P$ and $Q$ not being hashed to the same sequence is
bounded by 
\begin{align*}
  \Prob{\bigcup_{1\leq k \leq \ell}  E_k}{} \leq \sum_{1\leq k \leq \ell} \Prob{E_k}{} = 2dm\cdot\frac{\distFr{P}{Q}}{\delta} 
\end{align*}
 and the lemma follows.
\end{proof}

\begin{lemma}\lemlab{dfd:a2}
For any value of $\delta$ and for any $P,Q \in \Curves{d}$, if there exists a value of $t \in [0,\delta)^d$ such that 
$h^t_{\delta}(P)=h^t_{\delta}(Q)$, then it holds that $\distFr{P}{Q} \leq {\sqrt{d}}\cdot \delta.$
\end{lemma}

\begin{proof}
In the case that $h^t_{\delta}(P)=h^t_{\delta}(Q)$, it holds that $\distFr{h^t_{\delta}(P)}{h^t_{\delta}(Q)}= 0$.
Snapping a curve to the randomly shifted grid changes the position of each
vertex by at most $\frac{\sqrt{d}}{2}\cdot\delta$. Therefore, it holds that 
$\distFr{P}{h^t_{\delta}(P)} \leq  \frac{\sqrt{d}}{2}\cdot\delta$ and similarily 
$\distFr{Q}{h^t_{\delta}(Q)} \leq \frac{\sqrt{d}}{2}\cdot\delta$. By the triangle inequality,
\begin{align*}
\distFr{P}{Q} \leq \distFr{h^t_{\delta}(P)}{P} + \distFr{h^t_{\delta}(P)}{h^t_{\delta}(Q)} + \distFr{h^t_{\delta}(Q)}{Q} \leq \sqrt{d}\cdot{\delta}.
\end{align*}
\end{proof}

The next theorem  follows by plugging in the bounds of Lemmas~\ref{lem:dfd:a1} and~\ref{lem:dfd:a2}.
\begin{theorem}\thmlab{basic:dfd}
Let $P, Q \in \Curves{d}$ be two curves with $m_1$ and $m_2$ points, respectively, and let $m=\min\{m_1,m_2\}$, $\delta= 4dmr$ and  $c=4d^{\frac{3}{2}}m$.
It holds that: 
\begin{compactenum}[(i)]
\item if $\distFr{P}{Q} < r$, then $\Prob{h^t_{\delta}(P) = h^t_{\delta}(Q)}{\Hash^{\texttt{L}}_{\delta}} > \frac{1}{2}$;
\item if $\distFr{P}{Q} > cr$, then $\Prob{h^t_{\delta}(P) = h^t_{\delta}(Q)}{\Hash^{\texttt{L}}_{\delta}}= 0.$
\end{compactenum}
\end{theorem}
\begin{proof}
The first claim follows by plugging in the bounds of \lemref{dfd:a1}:
\begin{align*}
 \Prob{h^t_{\delta}(P) = h^t_{\delta}(Q)}{} > 1- \pth{2 d m \cdot \frac{\distFr{P}{Q} }{\delta}} > 1- \frac{\distFr{P}{Q}}{2r} > \frac{1}{2}.
\end{align*}
On the other hand, the second claim follows from \lemref{dfd:a2}:
\begin{align*} 
\distFr{P}{Q} > c\cdot r = 4 d^{3/2} m r = \sqrt{d}\cdot {\delta} \quad\Rightarrow\quad h^t_{\delta}(P) \neq h^t_{\delta}(Q).
\end{align*} 
\end{proof}

\section{Constant approximation factor}
\seclab{constant}
\seclab{const:all}
In the previous section we analyzed a very efficient LSH with linear
approximation factor. On the other end of the spectrum, we can also design an
LSH with constant approximation factor, but higher running time. Conceptually,
the easiest way to do this is to randomly and independently perturb the vertices of each curve and snap them to a fixed grid.

\subsection{Algorithm}
\seclab{const}
The described scheme is asymmetric. We assume that we have two types of curves,
which we call input curves and query curves.   
Consider an input curve $P=p_1, \dots,p_m$, and let $G_{\delta}$ be the canonical $d$-dimensional grid of resolution $\delta$ defined in the previous section.
Let $t_P=t_1,\dots,t_m$ be a sequence of independent random
variables which are uniformly  distributed in
$\pbrc{-\frac{\delta}{2},\frac{\delta}{2}}^d$.  
We perturb the vertices of $P$:
Let $P'=p'_1,\dots,p'_m$ be the perturbed curve with $p'_i=p_i+t_i$.
We snap the curve $P'$ to the grid ${G}_{\delta}$. More precisely, 
we replace each vertex $p'_i$ with its closest grid point 
$p''_i = \argmin_{q \in {G}_{\delta}} \|p'_i - q\|$ to obtain the curve $P''$. 
In the next step we remove consecutive duplicates in $P''$. 
That is, we remove the vertex $p''_i$ if it is identical to $p''_{i-1}$.  
We define  $h^{t_P}_{\delta}(P)$ to be the result of this algorithm.

For a query curve $Q$, the hash function is the same. However, a different random sequence $t_Q$ is used for randomly perturbing the curve.
We let
$\Hash^{\texttt{C}}_{\delta}$ denote the LSH scheme defined this way: namely, $\Hash^{\texttt{C}}_{\delta}$ contains all  pairs $(h^{t_P}_{\delta}, h^{t_Q}_{\delta})$, where vectors $t_P$ and $t_Q$ consist of entries independent and identically distributed in $\pbrc{-\frac{\delta}{2},\frac{\delta}{2}}^d$.

\subsection{Analysis}

\begin{lemma}\lemlab{dfdb:a1}
Let $P, Q \in \Curves{d}$ be two curves with $m_1$ and $m_2$ points, respectively. Let 
$m=\min\{m_1,m_2\}$ and let $M=\max\{m_1,m_2\}$. For any $\delta>0$, 
it holds that
\begin{align*}
\Prob{h^{t_P}_{\delta}(P)=h^{t_Q}_{\delta}(Q)}{\Hash^{\texttt{C}}_{\delta}}
\geq \pth{\frac{1}{2}}^{dm}\cdot \pth{\frac{1}{2}-\frac{\distFr{P}{Q}}{\delta}}^{dM}
\end{align*}
In particular, if $\delta > 4\distFr{P}{Q}$, then the probability is strictly lower bounded by $2^{-2d(m_1+m_2)}$.
\end{lemma}
\begin{proof}
Note that for $\distFr{P}{Q} \geq \frac{\delta}{2}$ the claim is trivially true. 
Therefore, assume that $\distFr{P}{Q} < \frac{\delta}{2}$.
For simplicity assume first that $d=1$.
We bound the probability that $P$ and $Q$ do not hash to the same sequence.
To this end, consider an optimal traversal $T$ of $P$ and $Q$ with respect to
the discrete \Frechet distance. By \lemref{components}, we can assume that 
$|T|\leq m_1+m_2$ and each component is a star. Let $\ell$ denote the number of components of $T$.
For $1 \leq k \leq \ell$ denote with $E_k$ the event that not all vertices of the
$k$-th component are snapped to the same grid point. 
Assume that the center of the $k$-th star is a vertex $p_i$ of $P$ and that the
other vertices of the component are vertices $q_j,\dots,q_{j+c_k}$ of $Q$. The
analysis for the case where the center is a vertex of $Q$ is analogous.
There must be a grid point in either one of the two intervals to the left and
to the right of $p_i$: $I_{l}=[p_i-\frac{\delta}{2}, p_i)$ and
$I_{r}=[p_i, p_i+\frac{\delta}{2})$.
We analyze the case that there is a grid point in $I_{r}$, the other case is analogous.
Let $X_i$ be the event that $p_i' \in I_{r}$. Since $t_P$ is uniformly random in 
$\pbrc{-\frac{\delta}{2},\frac{\delta}{2}}^{m_1}$, it holds that $\Prob{X_i}{} \geq \frac{1}{2}$.
Now, let $Y_j$ be the event that $q'_j \in I_r$. If $q_j$ was in $p_i$'s
component, then there are two cases. Either $q_j$ lies in $I_l$ or in $I_r$.
In the first case, we have 
\[\Prob{Y_i}{} \geq \frac{\frac{\delta}{2} - |p_i-q_j|}{\delta} \geq \frac{1}{2} - \frac{\dist{P}{Q}}{\delta},\]
and in  the second case we have $\Prob{Y_i}{} \geq \frac{1}{2}$.
We can bound the probability that all vertices in the $k$-th component snap to the same grid point 
\[\Prob{\overline{E_k}}{} \geq \Prob{X_i \cap Y_j \cap \dots \cap Y_{j+c_k}}{} 
\geq \frac{1}{2}\cdot \pth{\frac{1}{2}- \frac{\dist{P}{Q}}{\delta}}^{c_k}\]

If all components are preserved, then the two curves will hash to the same sequence, therefore
\begin{eqnarray*}
\Prob{h^{t_P}_{\delta}(P) = h^{t_Q}_{\delta}(Q)}{}
&\geq& \Prob{ \bigcap_{1\leq k \leq \ell} \overline{E_k}}{} 
\geq \prod_{1 \leq k \leq \ell} \Prob{\overline{E_k}}{} \\
&\geq& \prod_{1 \leq k \leq \ell} \frac{1}{2} \pth{\frac{1}{2}- \frac{\dist{P}{Q}}{\delta}}^{c_k} 
\geq \pth{\frac{1}{2}}^\ell  \pth{\frac{1}{2}- \frac{\dist{P}{Q}}{\delta}}^{m_1+m_2-\ell}. 
\end{eqnarray*} 
The last inequality follows since $\pth{\sum_{1\leq k\leq \ell}
c_k}=m_1+m_2-\ell$. Indeed, each center of a component can be charged to this component and 
the remaining vertices make up the sum of the leaves of all components. The lemma is now implied for $d=1$ observing that
$\ell\leq \min\{m_1,m_2\}$, as implied by \lemref{components}. 
We get the lemma for general $d$ by observing that the dimensions are independent.
\end{proof}

The following result then holds.
\begin{theorem}\thmlab{const:dfd}
Let $P, Q \in \Curves{d}$ be two curves with $m_1$ and $m_2$ points, respectively, and let  $\delta= 4dr$ and  $c=4d^{3/2}$.
It holds that
\begin{compactenum}[(i)]
\item if $\distFr{P}{Q} < r$, then $\Prob{h^{t_P}_{\delta}(P)=h^{t_Q}_{\delta}(Q)}{\Hash^{\texttt{C}}_{\delta}} > \pth{\frac{1}{2}}^{2d(m_1+m_2)}$;
\item if $\distFr{P}{Q} > cr$, then 
$\Prob{h^{t_P}_{\delta}(P)=h^{t_Q}_{\delta}(Q)}{\Hash^{\texttt{C}}_{\delta}} = 0.$
\end{compactenum}
\end{theorem}
\begin{proof}
The  theorem follows by plugging in the bounds of \lemref{dfdb:a1} and by using same arguments as in the proof of \lemref{dfd:a2}.
\end{proof}

\section{Trade-off between approximation factor and query time}
\seclab{improve}
\seclab{improved:all}
In the previous two sections we have seen  schemes with linear and constant
approximations.  We now suggest a scheme  exhibiting a trade-off between the
collision probability of near points and the approximation factor. 
The basic idea is to randomly partition the input curves and to concatenate the
outcome of the basic LSH (\secref{basic:all}) applied to the different parts of
the curves.  

\subsection{Algorithm}
The scheme is asymmetric. Again, we assume that we have two types of curves,
which we call {input} and {query curves}. The difference in how they are
handled lies in the way we create the partition.   For an input
curve $P=p_1,\dots,p_{m}$, we randomly sample a partition into $K$
subsequences. To this end, we denote a partition of $P$ with 
$\Partition{s}{P} = \pth{\subseq{P}{1},\dots, \subseq{P}{K}}$ where the
subsequences are defined by a monotone sequence $s \in [m]^{K-1} $ as
follows.
\begin{eqnarray*}
 \subseq{P}{1}= p_{1},\dots,p_{s_1}; \qquad
 \forall~ 1 < i < K~:~
 \subseq{P}{i}= p_{s_{i-1}},\dots,p_{s_{i}}; \qquad
 \subseq{P}{K}= p_{s_{K-1}},\dots,p_{m}.
\end{eqnarray*}
There are at most $\binom{m+K-1}{K-1}$ ways to partition a curve of length
$m$ in this way.  We denote with $\Partitions_K$ the family of all valid
partitions for a given $m$.  
Let $t=t_1,\ldots, t_K$ be a sequence of independent random values evenly distributed in $[0,\delta)^d$.
Once we have partitioned the input curve $P$ into $K$ (overlapping) subsequences, 
we apply the basic LSH to each individual subsequence and concatenate the
resulting curves:
\begin{align*}
g^{t,s}_{\delta, K}( P) 
= h^{t_1}_{\delta}\pth{\subseq{P}{1}} \oplus
h^{t_2}_{\delta}\pth{\subseq{P}{2}} \oplus \dots \oplus
h^{t_K}_{\delta}\pth{\subseq{P}{K}}.
\end{align*}

A query curve $Q=q_1,\dots,q_{m}$ is subdivided into $K$ equal-sized
subsequences (deterministically), where the last subsequence may be shorter and
two consecutive sequences overlap by one element.  We  denote with
$\Partition{*}{Q}$ this partitioning into equal-sized subsequences.  For query
curves, we define $ g^{t,*}_{\delta, K}(Q) $ to be the resulting curve given by
applying the basic LSH to each individual subsequence and concatenating the
resulting curves.

For any given $\delta>0$ and $K \geq 1$, we denote with
$\Hash^{\texttt{T}}_{\delta,K}$ the family of asymmetric hash functions created
this way: that is, $\Hash^{\texttt{T}}_{\delta,K}$ consists of  tuples
$(g_{\delta, K}^{t,s},g_{\delta, K}^{t,*})$ where the entries of $t$ are independently and
identically distributed  in $[0,\delta)^d$ and $\Partition{s}{P}$ is uniformly
chosen at random from $\Partitions_K$.

\subsection{Analysis}

We have the following theorem which generalizes \thmref{basic:dfd}. Using the
parameter $K$ we get a tradeoff between approximation factor and  query time.

%
\begin{theorem}\thmlab{improved:dfd}
  Let $P, Q \in \CurveSet$ be two curves with $m_1$ and $m_2$ points,
respectively, and let $M=\max\{m_1,m_2\}$. 
Let $K\geq1$ be a given integer and let  $\delta= 4dr \cdot \ceil{\frac{M}{K}}$ and
$c=4d^{\frac{3}{2}}\cdot\ceil{\frac{M}{K}}$.  It holds that 
\begin{compactenum}[(i)]
\item if $\distFr{P}{Q} < r$, then 
$\Prob{g_{\delta, K}^{t,s}(P)=g_{\delta, K}^{t,*}(Q)}{\Hash^{\texttt{T}}_{\delta,K}} \geq \pth{\frac{1}{2}}^{K} \cdot \binom{ M+K-1}{K-1}^{-1}$;
\item if $\distFr{P}{Q} > cr$, then 
$\Prob{g_{\delta, K}^{t,s}(P)=g_{\delta, K}^{t,*}(Q)}{\Hash^{\texttt{T}}_{\delta,K}} = 0.$
\end{compactenum}
\end{theorem}

\begin{proof}
We first prove (i). Let $T$ be an optimal traversal of $P$ and $Q$.
We say two partitions $\Partition{s}{P}$ and $\Partition{r}{Q}$ are 
\emph{consistent} with respect to $T$ if and only if $(s_i,r_i) \in T$ for all $1 \leq i \leq K-1$.
Let $E$ denote the event that the partition $\Partition{s}{P}$ used in the hash
functions is consistent with $\Partition{*}{Q}$ with respect to $T$. 
By construction this happens for at least one of the partitions in
$\Partitions_K$. Therefore, $\Prob{E}{} \geq \frac{1}{|\Partitions_K|}$.  Now,
let $E_i$ be the event that 
$h^{t_i}_{\delta}\pth{\subseq{P}{i}}=h^{t_i}_{\delta}\pth{\subseq{Q}{i}}.$
By \lemref{dfd:a1} we have that
\begin{align*}
\Prob{E_i~|~E}{} 
&\geq 1- \pth{2 d m' \cdot \frac{\distFr{\subseq{P}{i}}{\subseq{Q}{i}} }{\delta}} 
 \geq 1- \pth{2 d \ceil{\frac{M}{K}} \cdot \frac{\distFr{P}{Q}}{\delta}}
 \geq \frac{1}{2}
\end{align*}
Note that we can assume $m' \leq  \ceil{\frac{M}{K}}$ in  the above inequality, since
$m'$ is the length of the shorter of the two subsequences in the lemma. By
construction, the length of $\subseq{Q}{i}$ will be at most $\ceil{\frac{M}{K}}$. 

Since the values $t_i$ are chosen pairwise independent, we have  
\begin{align*}  
\Prob{g_{\delta, K}^{t,s}(P)=g_{\delta, K}^{t,*}(Q)}{\Hash^{\texttt{T}}_{\delta,K}}
\geq  \pth{\prod_{1\leq i\leq K} \Prob{E_i ~|~ E}{}} \cdot \Prob{E}{}
\geq   \pth{\frac{1}{2}}^{K} \cdot \frac{1}{|\Partitions_K|} 
\end{align*}  
Using $|\Partitions_{K}| \leq \binom{M+K-1}{K-1}$, the first part of the claim follows.

As for the second part of the claim,  we can use \lemref{dfd:a2} applied to the subsequences.
If there exists a partition of $P$, 
and there exist $t=t_1,\dots,t_K$, such that for all $0 \leq i \leq K$
$h^{t_i}_{\delta}\pth{\subseq{P}{i}}=h^{t_i}_{\delta}\pth{\subseq{Q}{i}}$,
then  it holds  by \lemref{dfd:a2} that 
$\distFr{\subseq{P}{i}}{\subseq{Q}{i}} \leq \sqrt{d} \cdot \delta.$
In this case, we can combine the traversals of the subsequences to a traversal
of the entire curves. This combined traversal has the same cost, therefore 
it follows that $\distFr{P}{Q} \leq \sqrt{d} \cdot \delta.$
Consequently, if 
\[\distFr{P}{Q} > cr = \frac{4d^{\frac{3}{2}}M}{K} r =  \sqrt{d} \cdot \delta,\]
then it cannot happen that  
$g_\delta^{t,s}(P)=g_\delta^{t,*}(Q)$
for any combination of $t=t_1,\dots,t_K$ and $s$.
\end{proof}

\begin{corollary}\corlab{improved:dfd:bound}
  Let $P, Q \in \CurveSet$ be two curves with $m_1$ and $m_2$ points,
respectively, and let $M=\max\{m_1,m_2\}$. 
Let $K\geq1$ be a given integer and let  $\delta= 4dr \cdot \ceil{\frac{M}{K}}$ and
$c=4d^{\frac{3}{2}}\cdot\ceil{\frac{M}{K}}$.  It holds that 
\begin{compactenum}[(i)]
\item if $\distFr{P}{Q} < r$, then 
$\Prob{g_{\delta, K}^{t,s}(P)=g_{\delta, K}^{t,*}(Q)}{\Hash^{\texttt{T}}_{\delta,K}} > \pth{\frac{1}{4}}^{K}\cdot \pth{\frac{1}{M}}^{K-1}$;
\item if $\distFr{P}{Q} > cr$, then 
$\Prob{g_{\delta, K}^{t,s}(P)=g_{\delta, K}^{t,*}(Q)}{\Hash^{\texttt{T}}_{\delta,K}} = 0.$
\end{compactenum}
\end{corollary}
\begin{proof}
The result just follows from the previous \thmref{improved:dfd} by observing that $ \binom{ M+K-1}{K-1}^{-1}\geq 1/(2M)^{k-1}$.
\end{proof}

%
%
%

\section{Handling constrained alignments}
\seclab{constraints}
We now focus on LSH for discrete \Frechet distance with constraints on the alignment. 
We first target the $w$-anchored distance in \secref{anchored}, and then the $w$-speed distance in \secref{speed}.
As in the previous sections, the schemes are asymmetric and consist of a partitioning of the curve into subsequences and on the application of the basic LSH scheme to each subsequence.
However, the partitions are different since they leverage on random processes on both input and query curves,  consecutive subsequences do not overlap, and the constraints are exploited.
We let $\ell \geq 1$ denote an arbitrary given integer that allows to trade-off the collision  probability of near curves with a bi-criteria approximation on the distance and on the anchored alignment.

\subsection{LSH for anchored distances}
\seclab{anchored}
Consider an input curve $P=p_1,\dots,p_{m}$ and let $r_P=r_{P,1}, r_{P,2},\ldots r_{P,m}$ and $t=t_1, t_2,\ldots, t_m$ denote sequences of independent and identically distributed random variables in $[1, w/2]$ and $[0,\delta)^d$ respectively, where $\delta$ is a suitable parameter defined later.
The partition of $P$ consists of a fixed partitioning into subsequences of length $\ell$, followed by a random perturbation of subsequence lengths. Specifically, the following three operations are performed:
\begin{compactenum}[(i)]
\item Partition $P$ into subsequences $\subseq{P'}{1},\dots, \subseq{P'}{K'}$ with $K'=\lceil m/\ell \rceil$ of size $\ell$, with the possible exception of the last subsequence. 
Let $s' \in [m]^{K'+1} $ be the vector denoting the final indexes of each subsequence, that is 
$\subseq{P}{i}= p_{\pth{s'_{i-1}+1}},\dots, p_{s'_{i}}$: we have $s'_0=0$, $s'_{K'}=m$ and $s'_{i}=i\ell$ for each $1\leq i < K'$.
\item Random perturb the final index of each subsequence with the random vector $r_P$: for each $1\leq i < K'$,  set $s_{i}=\min\{s_{i}+r_{p,2} , m\}$.
\item Clean the partition by removing overlaps among subsequences: for each $1 \leq i < K'$ and starting from $i=1$, remove each subsequence where $s'_{i}\leq s'_{j}$ for some $j<i$. 
We let  $\Partition{r_P}{P} = \pth{\subseq{P}{1},\dots, \subseq{P}{K}}$ denote the resulting  partition of $P$ with $K \leq \lceil m/\ell \rceil$ and let $s_{P} \in [m]^{K+1} $ be the resulting vector denoting the final indexes of each subsequence (note that each subsequence has now length at most $\ell+w$). 
\end{compactenum}

Once  curve $P$ has been partitioned into $K$ subsequences, we apply the basic LSH in \secref{basic:all} to each  subsequence using the random shifts given by sequence $t$.
Specifically, we snap the $i$-th subsequence $\subseq{P}{i}$ on a grid of side $\delta$ shifted by the random value $t_i$ and remove consecutive duplicates within each subsequence; the remaining values denote the hash value of $\subseq{P}{i}$ and we denote them with $h^{t_i}_{\delta}\pth{\subseq{P}{i}}$.
The final hash value $g^{t, r_P}_{w,\delta,\ell}(P)$  of curve  $P$ is the array containing the hash of each subsequence, specifically:
\begin{align*}
g^{t, r_P}_{w,\delta,\ell}(P) = 
\left( h^{t_1}_{\delta}\pth{\subseq{P}{1}}, 
h^{t_2}_{\delta}\pth{\subseq{P}{2}}, \dots,
h^{t_K}_{\delta}\pth{\subseq{P}{K}}\right).
\end{align*}
We observe that the final hash value is not a curve as in previous sections, but an array of curves. Equality between two curves then holds only if the two hash values have the same length and coincide in each position (i.e., the hash values $((a,b),(c))$ and $((a),(b,c))$ do not collide, but they collide if their are considered as a single curve $(a,b,c)$). This enforces the alignment constraint.

The hash process of a query curve $Q$ is the same: however, a different random sequence $r_Q$ is used to partition the curve, while the same sequence $t$ of random shifts is kept. Due to the different random bits in $r_Q$ the proposed LSH scheme is asymmetric.
We let $\Hash_{w,\delta,\ell}^{\texttt{A}}$  denote the hash family consisting of all possible pairs of hash functions $\left(g^{t, r_P}_{w,\delta,\ell}, g^{t, r_Q}_{w,\delta,\ell}\right)$.

The next \thmref{anchored} shows that the scheme has a bi-criteria approximation:
In addition to the distance approximation $c$, the scheme has also an approximation  on the alignment.
As an example, we observe that two curves with a $w$-anchored distance larger than $cr$ can still collide if they have  a $w+2(\ell-1)$-anchored distance lower than $cr$. 
In order to prove the theorem, we introduce the two following lemmas.

\begin{lemma}\lemlab{anchored_near}
Let $P, Q \in \CurveSet$ be two curves with $m_1$ and $m_2$ points,
respectively and let $m=\min\{m_1,m_2\}$. Let $w$ be the traversal width, $\ell\geq 1$ be an arbitrary integer, $\delta= 4dr\ell$.
If $\distAFr{P}{Q} < r$, then
$\Prob{g^{t, r_P}_{w,\delta,\ell}(P)=g^{t, r_Q}_{w,\delta,\ell}(Q)}{\Hash_{w,\delta,\ell}^{\texttt{A}}}
> \pth{{1}/{\sqrt{2}w }}^{2m/\ell}$.
\end{lemma}
\begin{proof}
 Let $T$ be an optimal $w$-anchored traversal of $P$ and $Q$, and let $c_1,\ldots c_v$ denote the $v$ non-overlapping components, and let $\pi_i$ and $\gamma_i$ denote the indexes of the components containing $p_i$ and $q_i$, respectively.
The two curves collide when the following two events happen for every $1\leq i \leq \lceil m/\ell \rceil$: 
\begin{itemize}
\item \emph{Event $E_{1,i}$:} $s_{P,i}=M_{P,i}$ and $s_{Q,i}=M_{Q,i}$, where  $M_{P,i}$ and $M_{Q,i}$ are the indexes of the rightmost vertices of $P$ and $Q$ in $c_{\pi_{i\ell-1}}$ and $c_{\gamma_{i\ell-1}}$.
\item \emph{Event $E_{2,i}$:} the hash values of the $i$-th subsequences of $P$ and $Q$ are the same (i.e., $h^{t_i}_{\delta}\pth{\subseq{P}{i}}=h^{t_i}_{\delta}\pth{\subseq{Q}{i}}$).
\end{itemize} 
Intuitively, the first event guarantees that the components of $T$ are not cut by the random partition, while the second event requires the basic LSH to work on each subsequence.
We show by induction that event $E_j = \bigcap_{i=1}^{j} (E_{1,i} \cap E_{2,i})$  (i.e., both events hold for the first $j$ subsequences) happens with probability at least $1/(\sqrt{2} w)^{2j}$, for $1\leq j < K=\lceil m/\ell\rceil$.

Assume $j=1$. By definition of $w$-anchored traversal, we have $i\ell \leq  M_{P,1}, M_{Q,1} \leq i\ell+w/2$.
Event $E_{1,1}$ happens with probability $1/w^2$, that is when the random shift moves the final indexes of $\subseq{P}{i}$  and $\subseq{Q}{i}$ in  $M_{P,i}$ and $M_{Q,1}$, respectively.
Conditioning on $E_{1,1}$,  event $E_{2,1}$ holds with probability at least $1/2$: indeed by \lemref{dfd:a1} we have that
\begin{align*}
\Prob{E_{2,1}|~E_{1,1}}{} 
&\geq 1- \pth{2 d \ell \cdot \frac{\distFr{\subseq{P}{i}}{\subseq{Q}{i}} }{\delta}} 
 \geq 1- \pth{2 d \ell \cdot \frac{\distAFr{P}{Q}}{\delta}}
 \geq \frac{1}{2}.
\end{align*}
Therefore, event $E_1$ happens with probability $1/(2 w^2)$.

Suppose now that $E_{j-1}$ holds with probability $1/(\sqrt{2} w)^{2(j-1)}$ and let $j<K$.
Since  $M_{P,j-1}$ and $M_{Q,j-1}$ are  the last indexes in segments $\subseq{P}{j-1}$ and $\subseq{Q}{j-1}$, we have $\Prob{E_{1,j}|~E_{j-1}}{} = 1/w^2$ by mimic the argument with $j=1$.
Further, we have that $\Prob{E_{2,j}|~E_{1,j}}{} \geq 1/2$. Therefore 
$\Prob{E_{j}|~E_{j-1}}{} =  \Prob{E_{1,j}|~E_{j-1}}{} * \Prob{E_{2,j}|~E_{1,j}}{} \geq 1/(\sqrt{2} w)^{2j}$. 

When $j=K$, we have that $\Prob{E_{1,K}|~E_{K-1}}{}=1$ since the end points  of $\subseq{P}{K}$ and $\subseq{Q}{K}$ are fixed on
$p_{m_1}$ and $Q_{m_2}$.
By the inductive assumption, the shortest sequence between $\subseq{P}{K}$ and $\subseq{Q}{K}$  contains less than $x=m\mod \ell$ vertices.
Then, we have 
\begin{align*}
\Prob{E_{2,K}|~E_{1,K}}{}  \geq 1- \pth{2 d x \cdot \frac{\distFr{\subseq{P}{i}}{\subseq{Q}{i}} }{\delta}}  \geq 1- \frac{x}{2 \ell}
 \geq \frac{1}{2^{x/\ell}}.
\end{align*}
Since $K = \lceil m/\ell\rceil$ and $x = m \mod \ell$, the two curves $P$ and $Q$ collide with probability 
$\Prob{E_{K-1}}{} \Prob{E_{2,K}}{}\geq 
(1/(2 w^2)^{K-1}) (1/2^{x/\ell}) \geq 1/(\sqrt{2} w)^{2 m/\ell}$.
\end{proof}

\begin{lemma}\lemlab{anchored_far}
Let $P, Q \in \CurveSet$ be two curves with $m_1$ and $m_2$ points,
respectively and let $m=\min\{m_1,m_2\}$.   Let $w$ be the traversal width, $\ell\geq 1$ be an arbitrary integer, $\delta= 4dr\ell$, and
$c=4d^{\frac{3}{2}} \ell$.
If $\distAFrx{P}{Q}{(w+2(\ell-1))} > cr$, then
$\Prob{g^{t, r_P}_{w,\delta,\ell}(P)=g^{t, r_Q}_{w,\delta,\ell}(Q)}{\Hash_{w,\delta,\ell}^{\texttt{A}}} = 0$.
\end{lemma}
\begin{proof}
By hypothesis, the two curves $P$ and $Q$ have $w\ell$-anchored \Frechet distance larger than $cr$, which implies there cannot be a $w\ell$-anchored traversal with cost smaller than or equal to $cr$. 
Assume that  $P$ and $Q$ collide under the described hashing scheme; then both curves have been split into $K$ subsequence and the hash values of $\subseq{P}{i}$ and $\subseq{Q}{i}$, for each $1\leq i \leq K$ collide.
By \lemref{dfd:a2}, there exists a traversal of cost at most $\sqrt{d}  \delta$ between $\subseq{P}{i}$ and $\subseq{Q}{i}$, for each $i$. Moreover, this traversal is a $w+2(\ell-1)$-anchored traversal since the partitioning guarantees that the indexes of vertices in $\subseq{P}{i}$ and $\subseq{Q}{i}$ differ by at most $w/2+\ell-1$.
This however implies that there exist a $w+2(\ell-1)$-traversal of $P$ and $Q$ of cost at most $\sqrt{d}  \delta = cr$, which is a contradiction.
Therefore, two curves with $(w+2(\ell-1))$-anchored discrete \Frechet distance cannot collide.
\end{proof}

\begin{theorem}\thmlab{anchored}
Let $P, Q \in \CurveSet$ be two curves with $m_1$ and $m_2$ points,
respectively and let $m=\min\{m_1,m_2\}$. Let $\ell\geq 1$ be an arbitrary integer, $\delta= 4dr\ell$, and
$c=4d^{\frac{3}{2}} \ell$.
Then, it holds that:
\begin{compactenum}[(i)]
\item if $\distAFr{P}{Q} < r$, then 
$\Prob{g^{t, r_P}_{w,\delta,\ell}(P)=g^{t, r_Q}_{w,\delta,\ell}(Q)}{\Hash_{w,\delta,\ell}^{\texttt{A}}}
> \pth{{1}/{\sqrt{2}w}}^{2m/\ell}$;
\item if $\distAFrx{P}{Q}{(w+2(\ell-1))} > cr$,  then 
$\Prob{g^{t, r_P}_{w,\delta,\ell}(P)=g^{t, r_Q}_{w,\delta,\ell}(Q)}{\Hash_{w,\delta,\ell}^{\texttt{A}}}= 0.$
\end{compactenum}
\end{theorem}
\begin{proof}
The theorem follows from Lemmas~\ref{lem:anchored_near} and~\ref{lem:anchored_far}.
\end{proof}

\subsection{LSH for bounded-speed distances}
\seclab{speed}
Consider an input curve $P=p_1,\dots,p_{m}$ and let $r_P=r_{P,1}, r_{P,2},\ldots, r_{P,m}$ and $t=t_1, t_2,\ldots t_m$ denote sequences of independent and identically distributed random variables in $[1,w \ell]$ and $[0,\delta)^d$ respectively.
We random partition curve $P$ into non overlapping subsequences of length given by the random sequence $r_{p}$.
Specifically, let $\Partition{s}{P} = \pth{\subseq{P}{1},\dots, \subseq{P}{K}}$ denote a partition of $P$ with $m/(w\ell) \leq K \leq m$ and let $s \in [m]^{K+1} $ be the vector denoting the initial and final indexes of a subsequence, that is 
$\subseq{P}{i}= p_{s_{i-1}+1},\dots,p_{s_{i}}$.
Then, $s$ satisfies the following conditions :
\begin{inparaenum}[(i)]
\item $s_0=0$ and $s_{K}=m$,
\item for $1\leq i \leq K$, $s_{i}- s_{i-1} = r_{p,1}$, which implies that  $s_{i}=\sum_{j=1}^{i} r_{p,1}$.
\end{inparaenum}
Once we have partitioned curve $P$ into $K$ subsequences, we continue as in the $w$-anchored LSH by applying the basic LSH to each  subsequence using the random shifts given by sequence $t$. 
For a query curve $Q$, the hash process  is the same, but a different random sequence $r_Q$ is used to partition the curve. 
We let $\Hash_{w,\delta,\ell}^{\texttt{S}}$ denote the hash family consisting of all possible pairs of hash functions $\left(g^{t, r_P}_{w,\delta,\ell}, g^{t, r_Q}_{w,\delta,\ell}\right)$.

The following \thmref{speed}  shows that the scheme has a bi-criteria approximation (note that the  alignment approximation in point $(ii)$ differs from the one for the anchored distance).  Its proof depends on the next two lemmas. 

\begin{lemma}\lemlab{speed_near}
Let $P, Q \in \CurveSet$ be two curves with $m_1$ and $m_2$ points,
respectively and let $m=\min\{m_1,m_2\}$. Let $w$ be the traversal width, $\ell\geq 1$ be an arbitrary integer, $\delta= 4dr\ell$.
If $\distSFr{P}{Q} < r$, then
$\Prob{g^{t, r_P}_{w,\delta,\ell}(P)=g^{t, r_Q}_{w,\delta,\ell}(Q)}{\Hash_{w,\delta,\ell}^{\texttt{S}}}
> \pth{{1}/{\sqrt{2}w \ell }}^{2m/\ell}$.
\end{lemma}
\begin{proof}
 Let $T$ be an optimal $w$-speed traversal of $P$ and $Q$, and let $c_1,\ldots c_v$ denote the $v$ non-overlapping components.
As for the anchored version, the two curves collide when the following two events happen for every $1\leq i \leq \lceil v/\ell \rceil$ (note a different definition for event $E_{1,i}$): 
\begin{itemize}
\item \emph{Event $E_{1,i}$:} all nodes of $P$ and $Q$ in components $c_{i\ell+1},\ldots, c_{(i+1)\ell}$ are contained in the $i$-th subsequence of $P$ and $Q$;
\item \emph{Event $E_{2,i}$:} the hash values of the $i$-th subsequences of $P$ and $Q$ are the same (i.e., $h^{t_i}_{\delta}\pth{\subseq{P}{i}}=h^{t_i}_{\delta}\pth{\subseq{Q}{i}}$).
\end{itemize} 
The proof  continues as in \lemref{anchored_near}, where the only difference is in the probability of event $E_{1,i}$. 
Indeed, we have that event $E_{1,1}$ happens with probability $1/(w\ell)^2$, that is when $P_1$ and $Q_1$ contain only the at most $w\ell$ nodes of $P$ and $Q$ in the first $\ell$ components.
The same probability holds for $E_{1,i}$, conditioning on $E_{1,i-1}$.
\end{proof}

\begin{lemma}\lemlab{speed_far}
Let $P, Q \in \CurveSet$ be two curves with $m_1$ and $m_2$ points,
respectively and let $m=\min\{m_1,m_2\}$.   Let $w$ be the traversal width, $\ell\geq 1$ be an arbitrary integer, $\delta= 4dr\ell$, and
$c=4d^{\frac{3}{2}} \ell$.
If $\distSFrx{P}{Q}{w\ell} > cr$, then
$\Prob{g^{t, r_P}_{w,\delta,\ell}(P)=g^{t, r_Q}_{w,\delta,\ell}(Q)}{\Hash_{w,\delta,\ell}^{\texttt{S}}} = 0$.	
\end{lemma}
\begin{proof}
By hypothesis, the two curves $P$ and $Q$ have $w\ell$-speed discrte \Frechet distance larger than $cr$, which implies there cannot be $w\ell$-speed traversal with cost smaller than or equal to $cr$. 
Assume that  $P$ and $Q$ collide under the described hashing scheme; then both curves have been split into $K$ subsequence and the hash values of $P_i$ and $Q_i$, for each $1\leq i \leq K$ collide.
By \lemref{dfd:a2}, there exists a traversal of cost at most $\sqrt{d}  \delta$ between $P_i$ and $Q_i$, for each $i$. Moreover, this traversal is a $w\ell$-traversal since each subsequence contains at most $w\ell$ nodes.
This however implies that there exist a $w\ell$-speed traversal of $P$ and $Q$ of cost at most $\sqrt{d}  \delta = cr$, which is a contradiction.
Therefore, two curves with $w\ell$-speed \Frechet distance cannot collide.
\end{proof}

\begin{theorem}\thmlab{speed}
Let $P, Q \in \CurveSet$ two curves with $m_1$ and $m_2$ points,
respectively and let $m=\min\{m_1,m_2\}$. Let $\ell\geq 1$ be an arbitrary integer, $\delta= 4dr\ell$, and
$c=4d^{\frac{3}{2}} \ell$.
Then, it holds that:
\begin{compactenum}[(i)]
\item if $\distSFr{P}{Q} < r$, then 
$\Prob{g^{t, r_P}_{w,\delta,\ell}(P)=g^{t, r_Q}_{w,\delta,\ell}(Q)}{\Hash_{w,\delta,\ell}^{\texttt{S}}}
> \pth{{1}/{\sqrt{2}w \ell}}^{2m/\ell}$;
\item if $\distSFrx{P}{Q}{w\ell} > cr$,  then 
$\Prob{g^{t, r_P}_{w,\delta,\ell}(P)=g^{t, r_Q}_{w,\delta,\ell}(Q)}{\Hash_{w,\delta,\ell}^{\texttt{S}}}= 0.$
\end{compactenum}
\end{theorem}
\begin{proof}
The proof follows from Lemmas~\ref{lem:speed_near} and~\ref{lem:speed_far}.
\end{proof}

\section{Extensions to dynamic time warping}
\seclab{ex:dtw}
All our schemes can be applied to DTW without any algorithmic change, and in this section we  analyze some of them.
We first investigate in \secref{basic:dtw} the basic scheme  in \secref{basic} for this distance.
Then, we provide a few insights on DTW with  constrained alignments in \secref{constrained_dtw}. 
We do not analyze the techniques proposed in Sections~\ref{sec:constant} and~\ref{sec:improve} since  they have the same linear approximation of the basic LSH, and---in contrast to our previous results for the \Frechet distance---do not provide a sublinear approximation for DTW.

\subsection{Analysis of the basic LSH}
\seclab{basic:dtw}
\begin{lemma}\lemlab{dtwd:a1}
Let $P, Q \in \Curves{d}$ be two curves with $m_1$ and $m_2$ points, respectively.
For any $\delta\geq 0$, it holds that
\begin{align*}
\Prob{h^t_{\delta}(P) = h^t_{\delta}(Q)}{\Hash^{\texttt{L}}_{\delta}}  \geq 1- \pth{d \cdot \frac{\distDTW{P}{Q} }{\delta}   }.
\end{align*}
\end{lemma}
\begin{proof}
Let $T$ be an optimal traversal of $P$ and $Q$ with respect to their DTW distance.
Let $\ell=|T|$ and denote with $d_k$ the distance $\|p_{i_k}-q_{j_k}\|$ for $1\leq k \leq \ell$.
We have that $\distDTW{P}{Q}=\sum_{1\leq k \leq \ell} d_k$.

Now, we bound the probability that $p_{i_k}$ and
$q_{j_k}$ for some fixed $(i_k,j_{k}) \in T$ are not snapped to the same grid
point. As in the proof of \lemref{dfd:a1}, two points are separated, if and only if they are separated
in at least one of their coordinate dimensions.  

Denote with $E_k$ the event that the pair is separated by the random shift $t$,
where $t$ is uniformly distributed in $[0,\delta)^d$.  By a union bound over the
pairs in $T$, we have that the probability of $P$ and $Q$ not being hashed to
the same curve is bounded by 
\begin{align*}
\hspace{-1em}
\Prob{\bigcup_{1\leq k \leq \ell}  E_k}{}
\leq \sum_{1\leq k \leq \ell} \Prob{E_k}{} 
\leq \sum_{1\leq k \leq \ell} d \cdot \frac{\|p_{i_k}-q_{j_k}\|}{\delta}
= \sum_{1\leq k \leq \ell} d \cdot \frac{d_k}{\delta} 
= d\cdot\frac{\distDTW{P}{Q}}{\delta} 
\end{align*}
and the lemma follows.
\end{proof}

\begin{lemma}\lemlab{dtwd:a2}
Let $P, Q \in \Curves{d}$ be two curves with $m_1$ and $m_2$ points, respectively,
and let $M=\max\{m_1,m_2\}$ and $\delta \geq 0$.
If there exists a value of $t \in [0,\delta)^d$ such that 
$h^t_{\delta}(P)=h^t_{\delta}(Q)$, then it holds that $\distDTW{P}{Q} \leq 2 M \sqrt{d} \cdot \delta.$
\end{lemma}
\begin{proof}
The DTW distance does not satisfy the triangle inequality,
however we can use a similar argument as in the proof of \lemref{dfd:a2}.
Assume that $h^t_{\delta}(P)=h^t_{\delta}(Q)$ is true from some $t \in [0,\delta)^d$
and denote $\widehat{P}=h^t_{\delta}(P)$ and $\widehat{Q}=h^t_{\delta}(Q)$.
Let  $\widehat{P}= \widehat{p_1},\dots,\widehat{p_s}$ and $\widehat{Q} =
\widehat{q_1},\dots,\widehat{q_s}$.  We have for $0 \leq r \leq s$ that
$\|\widehat{p_r}-\widehat{q_r}\|=0$.  Recall that $\widehat{P}$  (and
respectively, $\widehat{Q}$) was generated by snapping each point of $P$ to the
grid $\widehat{G}^t_{\delta}$ and by contracting each sequence of identical
vertices to one copy of the same vertex. 
We replace each $\widehat{p}_r$ by the original sequence of identical vertices 
$\widehat{p}_r^{(1)},\dots, \widehat{p}_r^{(a)}$ (and we do the same for each
$\widehat{q}_r$, denoting its sequence of identical vertices by
$\widehat{q}_r^{(1)},\dots, \widehat{q}_r^{(b)}$). Now, consider a traversal $T$ of
the resulting sequences that pairs each $\widehat{p}_r^{(1)}$ with
$\widehat{q}_r^{(1)}$. Furthermore, if $a>2$, then $T$ pairs each 
$\widehat{p}_r^{(i)}$ with $\widehat{q}_r^{(1)}$ and if $b>2$, then $T$ pairs
$\widehat{p}_r^{(a)}$ with each $\widehat{q}_r^{(i)}$. Finally, $T$ pairs
$\widehat{p}_r^{(a)}$ with  $\widehat{q}_r^{(b)}$, for each $r$. The resulting
traversal $T$ can also be applied to the original curves $P$ and $Q$. 
By the construction of $\widehat{P}$ and $\widehat{Q}$, any two paired vertices
have distance at most  $\sqrt{d}\cdot \delta$, since they snapped to the same
grid point.  It follows that
\[
\distDTW{P}{Q} = \min_{T' \in \TraversalSet} \sum_{(i_k,j_k) \in T'} \| p_{i_k} - q_{j_k}\|
 \leq \sum_{(i_k,j_k) \in T} \| p_{i_k} - q_{j_k}\|
 \leq  \sum_{(i_k,j_k) \in T} {\sqrt{d}}\cdot\delta  
 \leq  2  M {\sqrt{d}}  \delta, 
\]
where the last step can be obtained by observing that $|T| \leq m_1+m_2 \leq 2M$.
\end{proof}

\begin{theorem}\thmlab{basic:dtwd}
Let $P, Q \in \Curves{d}$ be two curves with $m_1$ and $m_2$ points, respectively, and let $M=\max\{m_1,m_2\}$, $\delta= 2dr$ and let $c=4d^{\frac{3}{2}}M$.
\begin{compactenum}[(i)]
\item if $\distDTW{P}{Q} < r$, then $\Prob{h^t_{\delta}(P) = h^t_{\delta}(Q)}{\Hash^{\texttt{L}}_{\delta}}  > \frac{1}{2}$;
\item if $\distDTW{P}{Q} > cr$, then $\Prob{h^t_{\delta}(P) = h^t_{\delta}(Q)}{\Hash^{\texttt{L}}_{\delta}}  = 0.$
\end{compactenum}
\end{theorem}
\begin{proof}
The proof follows by plugging in the bounds of $\delta$ in Lemmas~\ref{lem:dtwd:a1} and~\ref{lem:dtwd:a2}.
\end{proof}

\subsection{Handling constrained alignments}
\seclab{constrained_dtw}
The schemes in \secref{constraints} for  $w$-anchored/speed traversals automatically apply to DTW distance, with the same  collision probabilities stated in Theorems~\ref{theo:anchored} and~\ref{theo:speed}. 
However, the approximation factor is $4d^{3/2} (m_1+m_2)$, where $m_1$ and $m_2$ are curve lengths. 
The claim follows by mimicking the proofs for the \Frechet distance and use the bounds in \thmref{basic:dtwd}.
We provide the analysis only for the $w$-anchored DTW distance, being the one for $w$-speed DTW essentially the same.

\begin{lemma}\lemlab{anchored_near_DTW}
Let $P, Q \in \CurveSet$ be two curves with $m_1$ and $m_2$ points,
respectively and let $m=\min\{m_1,m_2\}$. Let $w$ be the traversal width, $\ell\geq 1$ be an arbitrary integer, $\delta= 2dr$.
If $\distADTW{P}{Q} < r$, then
$\Prob{g^{t, r_P}_{w,\delta,\ell}(P)=g^{t, r_Q}_{w,\delta,\ell}(Q)}{\Hash_{w,\delta,\ell}^{\texttt{A}}}
> \pth{{1}/{\sqrt{2}w }}^{2m/\ell}$.
\end{lemma}
\begin{proof}
The proof mimics the one for the  $w$-anchored \Frechet distance in \lemref{anchored_near}. 
The only  difference is in the value of $\delta$ required to get $\Prob{E_{2,i}|~E_{1,i}}{}\geq 1/2$, since, by \lemref{dtwd:a1},  the  hash values of the $i$-th subsequences of $P$ and $Q$ collide with probability at least $ 1- \pth{ d  \cdot \frac{\distFr{\subseq{P}{i}}{\subseq{Q}{i}} }{\delta}}$, which is independent of the number of components. 
\end{proof}

\begin{lemma}\lemlab{anchored_far_DTW}
Let $P, Q \in \CurveSet$ be two curves with $m_1$ and $m_2$ points,
respectively.   Let $w$ be the traversal width, $\ell\geq 1$ be an arbitrary integer, $\delta= 2dr$, and
$c=4d^{\frac{3}{2}} (m_1+m_2)$.
If $\distADTWx{P}{Q}{(w+2(\ell-1))} > cr$, then
$\Prob{g^{t, r_P}_{w,\delta,\ell}(P)=g^{t, r_Q}_{w,\delta,\ell}(Q)}{\Hash_{w,\delta,\ell}^{\texttt{A}}} = 0$.
\end{lemma}
\begin{proof}
The proof is almost the same of the one for \lemref{anchored_far} and the only difference is in the approximation provided by the basic LSH.
Indeed, by \lemref{dfd:a2}, there exists a traversal of cost at most $\sqrt{d} 2 M_i \delta$ between $\subseq{P}{i}$ and $\subseq{Q}{i}$ for each $i$, where $M_i$ denote the length of the longest sequence.
However, we have $\sum_{i=1}^{K} M_i \leq m_1+m_2$ and $\delta=2dr$.
Therefore the final approximation is $c=4 d^{3/2} (m_1+m_2)$.
\end{proof}

\begin{theorem}\thmlab{anchoredDTW}
Let $P, Q \in \CurveSet$ be two curves with $m_1$ and $m_2$ points,
respectively, and let $m=\min\{m_1,m_2\}$. Let $\ell\geq 1$ be an arbitrary integer, $\delta= 2dr$, and
$c=4d^{\frac{3}{2}} (m_1+m_2)$.
Then, the above hashing scheme guarantees that:
\begin{compactenum}[(i)]
\item if $\distADTW{P}{Q} < r$, then 
$\Prob{g^{t, r_P}_{w,\delta,\ell}(P)=g^{t, r_Q}_{w,\delta,\ell}(Q)}{\Hash_{w,\delta,\ell}^{\texttt{A}}}
> \pth{{1}/{\sqrt{2}w}}^{2m/\ell}$;
\item $\distADTWx{P}{Q}{(w+2(\ell-1))} > cr$,  then 
$\Prob{g^{t, r_P}_{w,\delta,\ell}(P)=g^{t, r_Q}_{w,\delta,\ell}(Q)}{\Hash_{w,\delta,\ell}^{\texttt{A}}} = 0.$
\end{compactenum}
\end{theorem}
\begin{proof}
The  theorem follows from Lemmas~\ref{lem:anchored_near_DTW} and~\ref{lem:anchored_far_DTW}.
\end{proof}

\section{Conclusion}
\seclab{concl}
To the best of our knowledge, this is the first paper providing LSH schemes for
curves.  When applied to the near neighbor problem, our techniques improve the
state of the art for the discrete \Frechet distance~\cite{i-approxnn-02}  under
different settings, and provide the first data structure with theoretical
guarantees for DTW.  The methods presented are simple enough that they may be
practical.
We do not know if our bounds are tight. It would be interesting to know 
if lower bounds can be obtained for the studied problem and/or to improve the upper bounds.
All of the presented LSH schemes exhibit the property that no collisions happen
between far points (i.e., $\alpha_2=0$).  An open question is to understand if
it is possible to slightly increase this collision probability  
(say $\alpha_2=1/n$) to get a better approximation factor. 
Another interesting direction would be to reduce space by exploiting the independence in the
approach described in \secref{const}, or by using a multiprobe approach~\cite{Lv07}.
Finally, we remark that our results only partially extend to DTW. As such, it
is still open to get a sublinear approximation for DTW. 
We hope that our work inspires further work in one of these directions.
%
%
\subparagraph*{Acknowledgements}
The authors would like to thank Rasmus Pagh and the anonymous reviewers for useful comments. This research was initiated at the Dagstuhl Seminar 16101 "Data Structures and Advanced Models of Computation on Big Data, 2016".

\bibliographystyle{abbrv}
\bibliography{frechetlsh}

\appendix
\section{Conditional lower bound}
\applab{clb}
There is a trivial reduction from the orthogonal vectors problem to the problem
of finding a close pair under the $\ell_{\infty}$ distance up to approximation $c<3$~\cite{Indyk01}. Since
$\ell_{\infty}$ embeds isometrically into  the \Frechet distance (while
preserving the dimension up to a constant factor), this implies a conditional
time lower bound for the near-neighbor problem under the \Frechet distance. 
In detail, the orthogonal vectors problem can be stated as follows. Given two sets of
vectors $A,B \subset \{0,1\}^d$ with $|A|=|B|=n$, does there exist a pair 
$a \in A$ and $b \in B$ such that $a$ and $b$ are orthogonal?
The orthogonal vectors conjecture, which can be related to the strong exponential
time hypothesis, states that for no $\eps>0$, there exists an algorithm for the
 orthogonal vectors problem that runs in time $O\pth{n^{2-\eps}\poly{d}}$ if $d>\log^2 n$.
This conjecture implies that there exists no data structure for exact near-neighbor
searching under the discrete \Frechet distance that achieves
both preprocessing time in $O\pth{n^{2-\eps}\poly{m}}$ and query time in
$O\pth{n^{1-\eps} \poly{m}}$ for any $\eps>0$. Indeed, if such a data structure would exist,
then we could solve an instance of orthogonal vectors by storing $A$ in 
this data structure and performing a query with each of the vectors of $B$.

\end{document}